\documentclass[11pt]{article}
\usepackage{hyperref}
\usepackage{url}
\usepackage{subcaption}
\usepackage{natbib}
\usepackage{multirow}
\usepackage{rotating}
\usepackage{array}
\usepackage{multicol}
\usepackage{float}
\usepackage{wrapfig}
\usepackage{cite}
\usepackage{xr}
\usepackage{algorithm}
\usepackage{algorithmic}
\usepackage{longtable}
\usepackage{amsmath, amssymb, mathtools, amsthm, url,amsbsy,enumerate}
\usepackage{color}

\newtheorem{theorem}{Theorem}[section]

\newtheorem{lemma}{Lemma}

\newtheorem{corollary}{Corollary}[section]

\newcommand{\ignore}[1]{}

\DeclarePairedDelimiter\floor{\lfloor}{\rfloor}
\DeclareMathOperator*{\argmin}{arg\min}

\begin{document}

\allowdisplaybreaks

\title{Estimating Partition-wise Models Using the Minimum Description Length Principle}
\title{Consistent Estimation for Partition-wise Models}
\title{Consistent Estimation for Partition-wise Regression and Classification Models}

\author{Rex C. Y. Cheung\footnote{Department of Statistics, University of California, Davis, One Shields Avenue, Davis, CA 95616, USA, emails: {\tt \{rccheung, aaue, tcmlee\}@ucdavis.edu}}
\and Alexander Aue$^*$
\and  Thomas C. M. Lee$^*$}
\date{\today}

\maketitle

\begin{abstract}
Partition-wise models offer a flexible approach for modeling complex and multidimensional data that are capable of producing interpretable results. They are based on partitioning the observed data into regions, each of which is modeled with a simple submodel.  The success of this approach highly depends on the quality of the partition, as too large a region could lead to a non-simple submodel, while too small a region could inflate estimation variance.  This paper proposes an automatic procedure for choosing the partition (i.e., the number of regions and the boundaries between regions) as well as the submodels for the regions.  It is shown that, under the assumption of the existence of a true partition, the proposed partition estimator is statistically consistent.  The methodology is demonstrated for both regression and classification problems.

\end{abstract}


\section{Introduction}

With the advent of complex data, partition-wise models have become one of the more important classes of data processing and statistical inference methods. They work by partitioning the data space into regions and assigning a simple model to each of  these regions. When comparing with fitting a single complicated non-linear model to the entire data space, partition-wise modeling comes with the advantage of creating simpler and potentially more sensible model interpretations. Depending on the partition specifiers and objective functions to optimize, different fitting methods have been proposed. One of the earliest methods is classification and regression tree (CART) \citep{Breiman-et-al84}, which recursively constructs a tree for the data space until all data within each region exhibits homogeneous behavior. More recent methods such as local supervised learning through space partitioning \citep{Wang-Saligrama} and cost-sensitive tree of classifiers \citep{Xu-et-al13} have been proposed to give further error analysis and improve model validation speed, respectively. Another method \citep{Oiwa-Fujimaki} partitions the data space into rectangular grids and estimates the specific linear models for each region individually. Lastly, \citet{Eto-Fujimaki-Morinaga-Tamano14} build a tree structure for the data space using a Bayesian approach to estimate the splitting locations. 

In terms of classification, support vector machines (SVMs) proposed by \citet{vl63} are one of the most widely used techniques in recent years.  A partition-wise variant known as locally linear support vector machines (LLSVM) is developed by \citet{torrladicky11}, which attempts to detect decision boundaries that are almost linear and classify data around that region using local linear classifiers.  Another variant is the local deep kernel learning of \citet{jgav13}, which is a tree-based classifier that has the goal of speeding up non-linear SVM prediction while maintaining accuracy. 

All of the method mentioned in the previous two paragraphs aim at improving the model fitness and the estimation algorithm to comprehend the complex and massive data structure.  However, many lack the theoretical justification on the statistical behavior of the individual partitioned regions and few theoretical results have been derived up to date. 
Among the results available in the literature, most are on recursion based regression. For example, \citet{gordonolshen78} and \citet{gordonolshen80} provide sufficient conditions on the estimator to be asymptotic Bayes risk efficient and $L^p$ consistent, respectively, and \citet{lugosinobel95} provide a similar result on histogram density estimation. More recently \citet{totheltinge11} consider data coming from complex sample designs. They propose a method that incorporates the information from a complex design when building regression trees, and establish sufficient conditions for asymptotic design $L^2$ consistency of these regression trees as estimator of the conditional mean of the population.

The goal of this paper is then to provide an automatic method for estimating the partition as well as the submodel for each partitioned region.  The method employs the minimum description length (MDL) principle \citep{Rissanen89,Rissanen07} to define the best fitting partition-wise model.  It is shown that, under some mild regularity conditions, the proposed MDL estimate converges to the true partition.  In other words, the proposed method consistently estimates the number of regions and the region boundaries.  To the authors' knowledge, this is the first theoretical result on the consistency of a partition estimate for partition-wise models.  It is an important result, as the quality of an estimated partition is crucial to the predictive power of an estimated partition-wise model.

The rest of this paper is organized as follow. Section~\ref{gen_inst} introduces the partition-wise models for the regression and classification settings. Section~\ref{MDLsection} presents the proposed method for estimating the partition-wise models. Section~\ref{theory} derives the theoretical results while Section~\ref{optimization} discusses the optimization techniques. Section~\ref{empiricalstudies} provides some empirical results and Section~\ref{conclusion} concludes.  Technical details can be found in the Appendix.

\section{Partition-wise Models for Regression and Classification}
\label{gen_inst}

This section defines the standard partition-wise linear models. 
For clarity, the simpler case in which only one predictor contains change points is presented first.

Assume that the observed data $(\textbf{x}^{\prime}_i, y_i), i = 1, \dots, n$, can be partitioned into $m + 1$ regions, and, for $l = 1, \dots, m+1$, the $l^{th}$ region can be modeled by a linear model with $s_l$ predictors. Note that $s_l$ can be different for different regions. Denote the total number of available predictors by $P$, so $\textbf{x}_i = (1, x_{i,1}, \dots x_{i,P})'$ and $s_l\leq P$ for all $l$.  For $r = 1, \dots, m$, let the $r^{th}$ change point be denoted by $k_r \in \mathbb{N}$, where $k_0 = 1$ and $k_{m+1} = n$; these change points define the boundaries at which adjacent regions meet.  The partition-wise linear model for the $l^{th}$ region is then
\begin{align}
\label{segreg}
y_i = \textbf{x}_i'\beta_l + \varepsilon_i, \quad k_{l-1} < x_{i, p} \leq k_l \quad \mbox{for some $p \in \{1, \dots P\}$},
\end{align}
where $\beta_l$ is the vector of coefficients for the $l^{th}$ region and $\varepsilon_i$'s are i.i.d.\ $N(0, \sigma^2)$. This partition-wise linear model assumes a natural ordering on the $p^{th}$ predictor, and the change points are determined with respect to this predictor. Thus model~(\ref{segreg}) is specified by the parameters $m$, $\mathcal{K} = \{k_1, \dots, k_m\}$, $\boldsymbol{\beta} = \{\beta_1', \dots, \beta_{m+1}'\}'$ and $\sigma^2$.  

Often times change points can occur at more than one predictor, thus one can extend (\ref{segreg}) in the following way. Denote by $\mathcal{B}$ (of size $B$) the set of predictors with change points; i.e. $\mathcal{B} \subseteq \mathcal{P}$, where $\mathcal{P}$ is the set of original predictors. For each $b \in \mathcal{B}$, denote by $k_{j_b,b}$ the $j_b^{th}$ change point for the $b^{th}$ predictor, and $l_b$ the number of change points for the $b^{th}$ predictor. Write $\mathcal{L} = \{l_b\colon b \in \mathcal{B}\}$ and $\mathcal{K}$ for the set of change locations. Similar to the definition of a region above, some particular subsets of $\mathcal{K}$ form partitions in the data space (denote by $\mathcal{R}_r$ for the $r^{th}$ region). Denote the set of these partitions by $\mathcal{R}$, and let $|\mathcal{R}| = R$. It is straightforward to see that $\bigcup_{r} \mathcal{R}_r$ is the smallest hypercube that covers the data space. Then, for each $r = 1, \dots, R$, the partition-wise linear model is
\begin{align}
\label{regionreg}
y_i = \textbf{x}'_i\beta_r + \varepsilon_i, \quad k_{j_b-1,b} < x_{i,b} \leq k_{j_b,b} \quad \mbox{for some $b \in \mathcal{B}$ and some $j_b \in \{1, \dots, l_b+1\}$}. 
\end{align}
Also, let $|\beta_r| = s_r$, which may vary across $r$. If $|\mathcal{B}| = 1$, then~(\ref{regionreg}) simplifies to~(\ref{segreg}). For the method to be proposed below, all parameters ($\mathcal{B}$, $\mathcal{L}$, $\mathcal{K}$, $\boldsymbol{\beta}$) are treated as unknown and have to be estimated.

It is important to point out that (\ref{regionreg}) is different from the tree structure of CART \citep{Breiman-et-al84}. The usual tree based structure partitions the data space recursively, where one splits an existing region into subregions without interfering with other existing regions. Model (\ref{regionreg}), however, forms a grid space in the data space. See Figure~\ref{fig:examplerealization} for a realization of this model.

\begin{figure}[ht]
\begin{center}
\includegraphics[scale = 0.2]{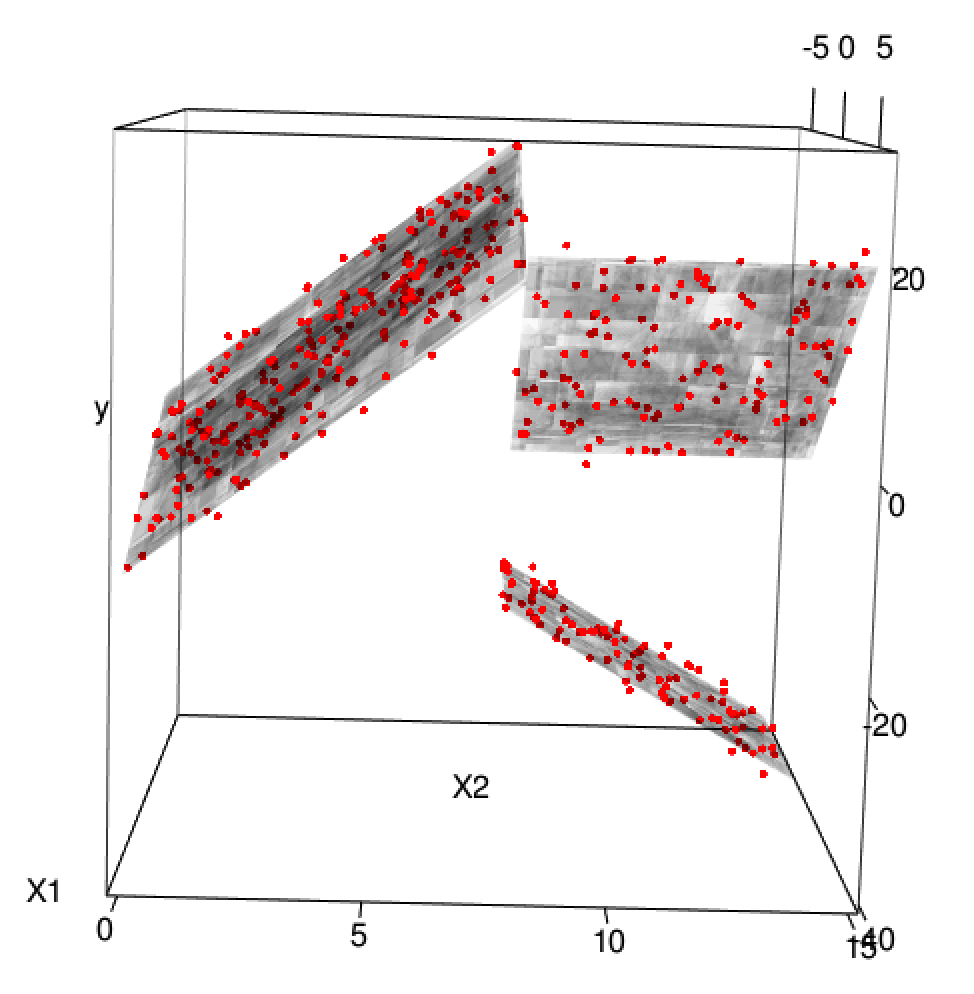}
\includegraphics[scale = 0.2]{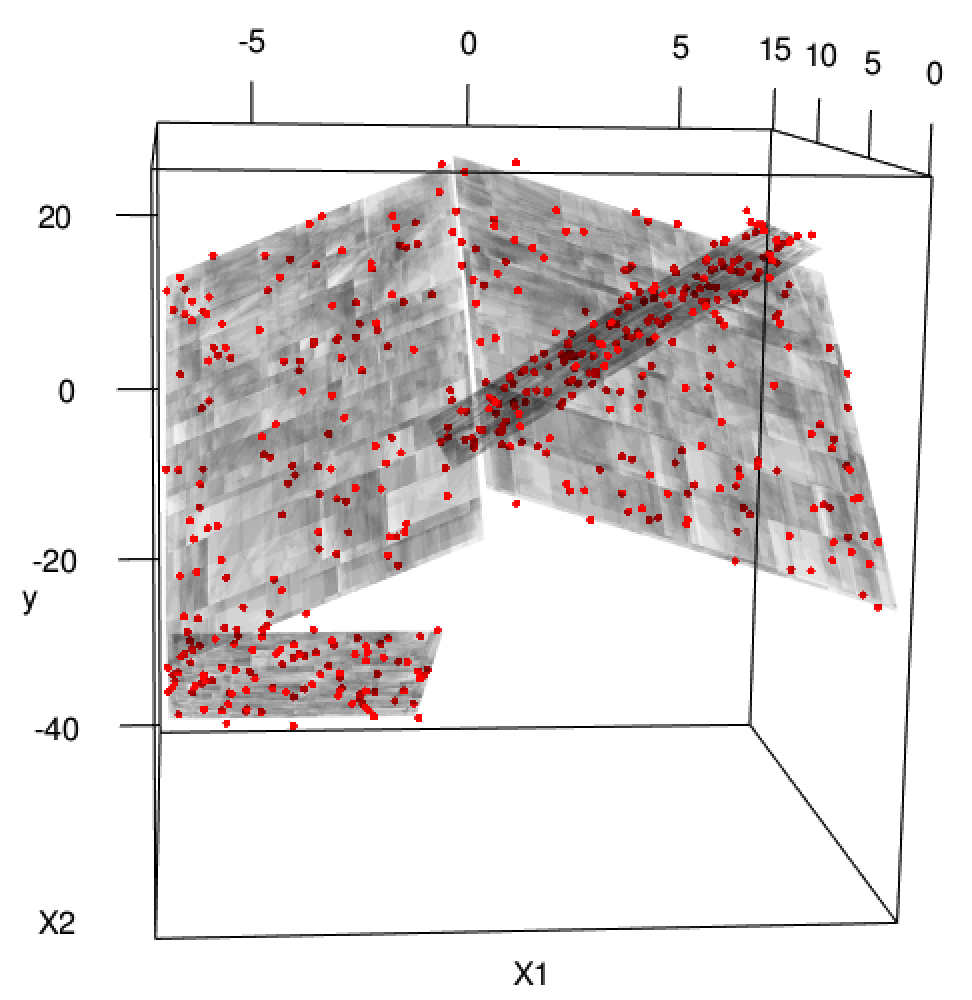}
\end{center}
\caption{A realization of model~(\protect\ref{regionreg}) viewed from two different angles. It contains two predictors, with change points at $-1$ for $x_1$ and 8 for $x_2$. The red dots are observations and the gray planes are the true signals.}
\vspace{-10pt}
\label{fig:examplerealization}
\end{figure}

Model~(\ref{regionreg}) can be extended to a more general setting using the generalized linear model
\begin{align}
\label{glm}
g(E[y_i|\textbf{x}_i]) = \textbf{x}_i'\beta_r, \quad k_{j_b-1,b} < x_{i,b} \leq k_{j_b,b},
\end{align}
where $g(\cdot)$ is a (known) link function. In particular, if one is interested in binary data, one could use
\begin{align}
\label{logisticmodel}
g^{-1}(\textbf{x}_i'\beta_r) = \frac{e^{\textbf{x}_i'\beta_r}}{1 + e^{\textbf{x}_i'\beta_r}}
\end{align}
or 
\begin{align}
\label{probitmodel}
g^{-1}(\textbf{x}_i'\beta_r) = \Phi(\textbf{x}_i'\beta_r) \quad \mbox{with $\Phi$ as the cdf of the standard normal distribution,}
\end{align}
which corresponds to the logistic and probit model, respectively.

\section{Change Point Detection and Variable Selection Using MDL}
\label{MDLsection}
\subsection{Linear Regression}
In standard linear regression theory, the estimate of $\beta$ can be obtained as
\begin{align*}
\hat{\beta} = \argmin \limits_{\beta \in \mathbb{R}^P} \sum_{i = 1}^n (y_i - \textbf{x}_i'\beta)^2.
\end{align*}
Under the normal assumption for $\varepsilon$, the least squares solution is the same as the maximum likelihood solution. Extension of this to the partition-wise linear model~(\ref{regionreg}) can be achieved as follows. As seen in~(\ref{regionreg}), the partition-wise linear model for regression is determined by the number of change points $l_b$ for each predictor in $\mathcal{B}$, the change locations $\mathcal{K} = (k_{1, 1}, \dots, k_{l_1, 1}, k_{1, 2}, \dots, k_{l_B, B})$, and the parameters for each region $\boldsymbol{\beta} = \{\beta_1, \dots, \beta_R\}$. If both $\mathcal{L}$ and $\mathcal{K}$ are known, then $\boldsymbol{\beta}$ can be estimated by solving
\begin{align*}
\boldsymbol{\hat{\beta}} = \argmin \limits_{\boldsymbol{\beta}} \sum_{\mathcal{R}_r \in \mathcal{R}} \sum_{i \in \mathcal{R}_r} (y_i - \textbf{x}_i'\beta_r)^2.
\end{align*}
Note that in practical minimization, it is assumed that the data is organized in a way that the data points of each region are grouped together.

The estimation of $\mathcal{L}$ and $\mathcal{K}$ is, however, less trivial, as different $\mathcal{L}$ will give different $\mathcal{K}$; i.e., the model dimensions are different. The rest of this section will apply the MDL principle to derive an estimate for $\mathcal{L}$ and $\mathcal{K}$.

The MDL principle is a model selection criterion. When applying the MDL principle, the ``best'' model is defined as the one that allows the greatest compression of the data $\textbf{y} = (y_1, \dots, y_n)$. That is, the ``best'' model enables us to store the data in a computer with the shortest code length. There are a few versions of MDL, and the ``two-part'' version will be used here.  The first part encodes the fitted model being considered, denoted by $\hat{\mathcal{F}}$, and the second part encodes the residuals left unexplained by the fitted model, denoted by $\hat{\mathcal{E}} = \hat{\textbf{y}} - \textbf{y}$, where $\hat{\textbf{y}}$ is the fitted value for $\textbf{y}$. Denote by $\text{CL}_{\mathcal{F}}(\textbf{y})$ the code length of $\textbf{y}$ under the model $\mathcal{F}$, then 
\begin{align}
\label{codelength}
\text{CL}_{\mathcal{F}}(\textbf{y}) = \text{CL}_{\mathcal{F}}(\hat{\mathcal{F}}) + \text{CL}_{\mathcal{F}}(\hat{\mathcal{E}}|\hat{\mathcal{F}}).
\end{align} 
The goal is to find the model $\hat{\mathcal{F}}$ that minimizes (\ref{codelength}). To use (\ref{codelength}), the two terms on the right need to be calculated. To fit the model, one should first determine which of the $P$ predictors contain change points, and the number of change points $l_b$ for each predictor $b$. Also, let $n_{j_b,b} = k_{j_b, b} - k_{j_b-1, b}$ be the number of observations between any two change points $j_b-1$ and $j_b$ for some predictor $b$. Since $\hat{\mathcal{F}}$ is completely characterized by $\mathcal{B}$, $\mathcal{L}, \mathcal{K}$ and $\boldsymbol{\beta}$, the code length of $\hat{\mathcal{F}}$ can be decomposed as
\begin{align}
\label{modelcode}
\text{CL}_{\mathcal{F}}(\hat{\mathcal{F}}) &= \text{CL}_{\mathcal{F}}(\mathcal{B}) + \text{CL}_{\mathcal{F}}(\mathcal{L}) + \text{CL}_{\mathcal{F}}(\mathcal{K}) + \text{CL}_{\mathcal{F}}(\boldsymbol{\beta}) \nonumber \\
&= \text{CL}_{\mathcal{F}}(\mathcal{B}) + \text{CL}_{\mathcal{F}}(l_1) + \dots + \text{CL}_{\mathcal{F}}(l_B) + \text{CL}_{\mathcal{F}}(n_{1,1}) + \dots \nonumber \\
& \quad+ \text{CL}_{\mathcal{F}}(n_{j_B+1, B}) + \text{CL}_{\mathcal{F}}(\beta_1) + \dots + \text{CL}_{\mathcal{F}}(\beta_R).
\end{align}
According to \citet{Rissanen89}, it requires approximately $\log_2I$ bits to encode an integer $I$ if the upper bound is unknown, and $\log_2I_u$ bits if $I$ is bounded from above by $I_u$. To encode $\mathcal{B}$, one needs to determine which of the $P$ predictors are selected. This takes $B\log_2P$ bits, as each of the $B$ predictors can be identified by an index upper bounded by $P$. For each set of $\{l_b, n_{1, b}, \dots, n_{l_b, b}\}$, ${b = 1, \dots, B}$, one needs to first decide which of the $B$ predictors are chosen (with code length $\log_2(B+1)$), then encode $l_b$ and $n_{1, b}, \dots, n_{l_b, b}$; therefore 
\begin{align*}
\text{CL}_{\mathcal{F}}(l_b) + \text{CL}_{\mathcal{F}}(n_{1,b}) + \dots + \text{CL}_{\mathcal{F}}(n_{l_b, b}) = \log_2(B + 1) + \log_2(l_b + 1) + \sum \limits_{z = 1}^{l_b + 1} \log_2n_{z, b}. 
\end{align*}
Note the $1$ is added in the first two $\log$ terms for computational purposes. Lastly, by \citet{Rissanen89}, it takes $\frac12\log_2N$ bits to encode a maximum likelihood estimate of a parameter computed from $N$ observations. To encode $\beta_r$, one needs to determine which region $\beta_r$ belongs to, which takes $\log_2R$ bits, and since $|\beta_r| = s_r$, 
\[\text{CL}_{\mathcal{F}}(\beta_r) = \log_2R + \frac{s_r}{2}\log_2n_r,\]
where $n_r$ is the number of observations in $\mathcal{R}_r$. Putting everything together, (\ref{modelcode}) becomes
\begin{align}
\label{modelcode2}
\text{CL}_{\mathcal{F}}(\hat{\mathcal{F}}) &= B\log_2P + \sum_{b \in \mathcal{B}}\bigg[\log_2(B + 1) + \log_2(l_b + 1) + \sum \limits_{z = 1}^{l_b + 1} \log_2n_{z, b}\bigg] + \sum_{\mathcal{R}_r \in \mathcal{R}}\Big(\log_2R + \frac{s_r}{2}\log_2n_r\Big). 
\end{align}

To obtain the second term of (\ref{codelength}), one can use the result of \citet{Rissanen89} that the code length of the residuals $\hat{\mathcal{E}}$ is the negative of the log-likelihood of the fitted model $\hat{\mathcal{F}}$. With the assumption $\varepsilon$ follows $N(0,\sigma^2)$,
\begin{align}
\label{errorcodelength}
\text{CL}_{\mathcal{F}}(\hat{\mathcal{E}}|\hat{\mathcal{F}}) = \frac{n}{2}\log(\hat{\sigma}^2) \qquad \mbox{with} \qquad
\hat{\sigma}^2 = \frac1n\sum\limits_{\mathcal{R}_r \in \mathcal{R}}\sum \limits_{i \in \mathcal{R}_r}(y_i - \textbf{x}_i'\hat{\beta}_r)^2. 
\end{align}
Combining (\ref{modelcode2}) and (\ref{errorcodelength}), the proposed MDL criterion for the best fitting partition-wise linear model for regression is
\begin{align}
\mbox{MDL}_{\rm reg}(\mathcal{B}, \mathcal{L}, \mathcal{K}, \boldsymbol{\beta}) 
=& B\log_2P + \sum_{b \in \mathcal{B}} \bigg[\log_2(B + 1) + \log_2(l_b + 1) + \sum \limits_{z = 1}^{l_b + 1} \log_2n_{z, b}\bigg] \nonumber \\
 &+ \sum_{\mathcal{R}_r \in \mathcal{R}}\Big[\log_2R + \frac{s_r}{2}\log_2n_r\Big] + \frac{n}{2}\log\Bigg(\frac1n\sum\limits_{\mathcal{R}_r \in \mathcal{R}}\sum \limits_{i \in \mathcal{R}_r}(y_i - \textbf{x}_i'\hat{\beta}_r)^2\Bigg).
\label{MDLregression}
\end{align}

\subsection{Logistic and Probit Regression}
The MDL criterion for logistic regression can be derived by replacing the last term in~(\ref{MDLregression}) with the negative log-likelihood of model~(\ref{glm})--(\ref{logisticmodel}), which gives   
\begin{align}
\label{MDLlogistic}
\mbox{MDL}_{\rm bin}(\mathcal{B}, \mathcal{L}, \mathcal{K}, \boldsymbol{\beta}) 
&= B\log_2P + \sum_{b \in \mathcal{B}}\Bigg[\log_2(B + 1) + \log_2(l_b + 1) + \sum \limits_{z = 1}^{l_b + 1} \log_2n_{z, b}\Bigg] \nonumber \\
& \hspace*{-1cm} + \sum_{\mathcal{R}_r \in \mathcal{R}}\Big[\log_2R + \frac{s_r}{2}\log_2n_r\Big] - \sum\limits_{\mathcal{R}_r \in \mathcal{R}}\sum \limits_{i \in \mathcal{R}_r}\Big[y_i\textbf{x}_i'\hat{\beta}_r - \log\big(1 + e^{\textbf{x}_i'\hat{\beta}_r}\big)\Big].
\end{align}
A similar expression can be derived for the probit regression.

\section{Large-Sample Theory}
\label{theory}
In order to assess the asymptotic behavior of the MDL criteria derived above, a few assumptions are needed on the underlying random process. First it is assumed that the observations $\{(\textbf{x}_i', y_i)\}_{i=1}^{n}$ follow model~(\ref{regionreg}), (\ref{logisticmodel}) or~(\ref{probitmodel}). Also, it is assumed that there exists a constant $M_0$ such that $|\textbf{x}_i|_{\infty} \leq M_0$ for all $i$. This assumption is needed to ensure the rate of growth of the predictors are restricted. Denote the true parameters with the superscript ``0''. Therefore the true number of change points for the $b^{th}$ predictor is written as $l_b^0$, and the corresponding change points are denoted as $\{k^0_{1,b}, \dots, k^0_{l_b^0, b}\}$. Denote the relative locations of the change points by $k^0_{j_b, b} = \floor{\lambda^0_{j_b,b}n}$, and $0 < \lambda^0_{1, b} < \dots < \lambda^0_{l_b^0,b} < 1$, for all $b \in \mathcal{B}$. To guarantee sufficient data for consistent parameter estimation within each region, it is assumed that there exists a $\delta > 0$ such that $\delta \ll \min \limits_{j_b} (\lambda^0_{j_b, b} - \lambda^0_{j_b-1,b})$. The collection of all potential change locations can then be written as 
\begin{align}
\label{breakcandidate}
\mathbf{\Lambda} &= \Big\{(\lambda_{1, b}, \dots, \lambda_{l_b,b}) \colon b \in \mathcal{B},\, 0 < \lambda_{1, b} < \dots < \lambda_{l_b^0,b} < 1,\, \delta \ll \min \limits_{j_b} (\lambda_{j_b, b} - \lambda_{j_b-1,b})\Big\}. 
\end{align}
Then, for the regression setting, the estimation of the parameters $\mathcal{B}^0$, $\mathcal{L}^0$, $\Lambda^0$ and $\boldsymbol{\beta}^0$ can be found by minimizing the MDL criterion:
\begin{align}
\label{regestimate}
(\hat{\mathcal{B}}, \hat{\mathcal{L}}, \hat{\Lambda}, \hat{\boldsymbol{\beta}}) = \argmin_{\mathcal{B}, \mathcal{L}, \Lambda,\boldsymbol{\beta} \in \mathcal{M}} \frac2n\mbox{MDL}_{\rm reg}(\mathcal{B}, \mathcal{L}, \Lambda, \boldsymbol{\beta}),
\end{align}
where $\mathcal{M} = \{(\mathcal{B}, \mathcal{L}, \Lambda, \boldsymbol{\beta})\colon \mathcal{B} \subseteq \mathcal{P}, \Lambda \in \mathbf{\Lambda}\}$. Similarly, one can set up a mirror version for the logistic/probit model setting by minimizing the MDL criterion:
\begin{align}
\label{logestimate}
(\hat{\mathcal{B}}, \hat{\mathcal{L}}, \hat{\Lambda}, \hat{\boldsymbol{\beta}}) = \argmin_{\mathcal{B}, \mathcal{L}, \Lambda,\boldsymbol{\beta} \in \mathcal{M}} \frac1n\mbox{MDL}_{\rm bin}(\mathcal{B}, \mathcal{L}, \Lambda, \boldsymbol{\beta}),
\end{align}
where $\mbox{MDL}_{\rm bin}(\cdot)$ is the MDL criterion~(\ref{MDLlogistic}) for binary data with logistic regression.  With this set up, the following theorems are proved in the Appendix.
\begin{theorem}
\label{breaklocthm}
Suppose the number of change points $\mathcal{L}^0$ is known. Then estimating the partition-wise linear model~(\ref{regionreg}) leads to 
\[
\hat{\Lambda} \to \Lambda^0 \quad \text{with probability one} \quad (n \to \infty),
\]
where $\hat{\Lambda}$ is the minimizer of the MDL criterion~(\ref{regestimate}). Similarly, estimating the partition-wise logistic model~(\ref{logisticmodel}) leads to
\[
\hat{\Lambda} \to \Lambda^0 \quad \text{with probability one} \quad (n \to \infty),
\]
where $\hat{\Lambda}$ is the minimizer of the MDL criterion~(\ref{logestimate}).  A parallel result holds for the probit model~(\ref{probitmodel}) with a MDL criterion designed for the probit link.
\end{theorem}

The above assumption that $\mathcal{L}^0$ is known is generally not practical. However, the theoretical result for the case of unknown $\mathcal{L}^0$ is much harder to work with, and is usually restricted to some special cases on the $\varepsilon_i$'s, such as the normal assumption.
\begin{theorem}
\label{breakconsistthm}
Suppose $\{(\textbf{x}_i',y_i)\}_{i=1}^{n}$ follow the partition-wise linear model~(\ref{regionreg}) and $\{\varepsilon_i\}_{i=1}^{n}$ are normally distributed. Furthermore, assume the set of predictors with change points $\mathcal{B}^0$ is known. Then the minimizers $(\hat{\mathcal{L}}, \hat{\Lambda})$ of the MDL criterion~(\ref{regestimate}) satisfy
\[
\hat{\mathcal{L}} \to \mathcal{L}^0 \quad \text{with probability one} \quad (n \to \infty),
\]
and
\[
\hat{\Lambda} \to \Lambda^0 \quad \text{with probability one} \quad (n \to \infty).
\]
\end{theorem}
With Theorem~\ref{breakconsistthm}, a similar result for probit model can be established:
\begin{corollary}
Suppose $\{(\textbf{x}_i',y_i)\}_{i=1}^{n}$ follow the probit model~(\ref{probitmodel}) and $(\hat{\mathcal{L}}, \hat{\Lambda})$ are the minimizers of an MDL criterion designed for the probit link. Then a similar result of Theorem~\ref{breakconsistthm} holds.
\end{corollary}
An inspection of the proof shows that the consistency result depends highly on the tail behavior of the likelihood function, one that the logistic model does not possess.  Theoretical development for the proof of this case is beyond the scope of this paper, although simulations suggest that the proposed MDL criterion for logistic regression correctly identifies $\mathcal{B}$ and $\mathcal{L}$ in the finite sample scenarios that were tested.

\section{Practical Minimization of the MDL Criterion}
\label{optimization}
The solution to the optimization problem of~(\ref{regestimate}) or~(\ref{logestimate}) involves a mix of discrete and continuous variables, where one needs to first determine the change points and their locations, then do variable selection and parameter estimation within each region.  
This section develops a fitting algorithm that combines univariate change point detection and binary particle swarm optimization (BPSO) to tackle this problem.

\subsection{Change Point Detection}
The first step of the algorithm is to find a candidate set of change points for each predictor in the data. Denote this set as $\mathcal{K}_{\rm sup}$, and its size by $K_{\rm sup}$. Note that ideally $\mathcal{K} \subseteq \mathcal{K}_{\rm sup}$.  
To locate the candidate set, one can use the proposed MDL criterion ((\ref{MDLregression}) or~(\ref{MDLlogistic})) with the following procedure.
For the first predictor, one first locates the change location that minimizes the MDL criterion. Conditioning on this first existing change location, one can locate a second change location similarly. Note that this assumes change points only occur at one predictor. To guarantee there is sufficient data for parameter estimation within each region, a minimum span constraint is imposed between any two change locations. Repeat this procedure for all predictors until a maximum number of change locations is reached, or until any of the candidate locations violate the minimum span constraint. Once the set $\mathcal{K}_{\rm sup}$ is determined, BPSO will be used to refine the set $\mathcal{K}$ that minimizes~(\ref{MDLregression}) or~(\ref{MDLlogistic}).

\subsection{Selection of Change Points via BPSO}
Particle swarm optimization (PSO) is a metaheuristic that aims to optimize a problem by iteratively improving a set of candidate solutions (swarm) with respect to a criterion function (here, the MDL criterion). The algorithm was first developed by \citet{KE95}, inspired by the social behavior of bird flocking or fish schooling. The algorithm shares some similarities with most evolutionary optimization algorithms such as genetic algorithms. In each, the algorithm initializes an initial population of candidate solutions and searches for the optimal solution by updating generations. 

PSO first initializes a set of candidate solutions. At each generation, each individual solution (particle) is updated according to a velocity value that suggests the current solution to converge towards the current global optimal solution. Each particle then keeps track of the best solution it has achieved so far (called $pbest$), as well as the best solution out of all the generated solutions (called $gbest$).  The algorithm iterates through these steps and stops when a certain stopping criterion is met. 

The original PSO was designed to solve optimization problems in continuous space. For discrete solutions (in particular, binary solutions), one needs to use a modified version, called BPSO \citep{KE97,Shen04}.  Roughly speaking, the BPSO proceeds in the same manner as PSO, but the solutions are now restricted to take values of 0 and 1, while the velocities are restricted to values between 0 and 1 through a sigmoid transformation. To apply BPSO to select the number of change points and change locations, $N$ particles are generated to encode the locations of candidate change points, as followed: each particle $X_i, i = 1, \dots, N$, can be expressed as a matrix of dimension $P$ by $n$, with element values 
\[X_{i,jk} = \left \{ \begin{array}{rl}
1, & \text{if the }k^{th}\text{ value of variable }j\text{ is a break candidate},\\
0, & \text{otherwise}.
\end{array}\right.
\]
In practice, one needs to ensure there are sufficiently many data points within each induced partition for sensible parameters estimation. For the current implementation, the particle should give partitions that contain at least $P$ observations in all induced partitions. Regenerate the particle if this constraint is not met.

\subsubsection{Initial Population}
For the current problem, the swarm initialization is done in the following way. The first particle ($X_1$) encodes all the break candidates in the set $\mathcal{K}_{\rm sup}$. The remaining particles are split into two sets, each with a different generating mechanism: each particle in the first half encodes a random subset of $\mathcal{K}_{\rm sup}$, while each particle in the second half encodes a random subset of $\mathcal{K}_{\rm sup}$ plus a random adjustment on the selected change locations. From experiments, $\mathcal{K}_{\rm sup}$ usually captures all the correct change points with locations very close to the true locations. The suggested adjustments are used to increase the solution search space and refine the estimated solutions. Once all the particles are generated, evaluate all particles (including variable selection described in Section~\ref{featselect}). Set $pbest$ of each particle as itself, and $gbest$ the particle which gives the smallest MDL value.

Initial velocities are also generated at this step. In BPSO, velocity can be interpreted as the probability that an element within the particle will take the value 1. The set of velocities has the same dimension as the swarm; i.e., has $N$ elements, each with dimension $P$ by $n$. All velocities are initialized to be 0.

\subsubsection{Update}
At each iteration $t$, the first step is to update the velocity. For each particle $i$, each element of velocity $v_i$ is updated via the formula 
\begin{eqnarray}
\label{velocity}
v_{i,jk}^{t} = |\omega v_{i,jk}^{t-1} + c_1 \times r_1 \times (pbest_{i,jk} - X_{i,jk}^{t-1}) + c_2 \times r_2 \times (gbest_{i,jk} - X_{i,jk}^{t-1})|,
\end{eqnarray}
where $\omega, c_1$, and $c_2$ are tuning parameters, and $r_1, r_2 \sim U(0,1)$. Using recommended settings, $\omega, c_1$, and $c_2$ are set to 1, 2, and 2 respectively. Note there is no guarantee that this updated velocity will lie between 0 and 1, thus a sigmoid transformation is applied, and the final velocity is 
\begin{eqnarray}
\label{vsigmoid}
v_{i,jk}^{t} = \frac{1}{1+e^{-v_{i,jk}^{t}}}.
\end{eqnarray}
Once the velocities are updated, the particles can be updated using the rule suggested in \citet{Shen04}:
\begin{equation}
\label{particle}
X_{i,jk}^{t} = \left \{ \begin{array}{rl}
X_{i,jk}^{t-1}, & \text{if  }v_{i,jk}^{t} \leq a,\\
pbest_{i,jk}, & \text{if } a < v_{i,jk}^{t} \leq \frac{1}{2}(1+a),\\
gbest_{jk}, & \text{if } \frac{1}{2}(1+a) < v_{i,jk}^{t} \leq 1\end{array}\right.
\end{equation}
for some $a \in (0,1)$. The current implementation uses $a = 0.5$. Lastly, $pbest_i, i = 1, \dots, N$, and $gbest$ are updated by comparing the MDL values of the new particles with the MDL values of the old particles. For each particle $i$, $pbest_i = X_i^{t}$ if $\mbox{MDL}(X_i^{t}) < \mbox{MDL}(X_i^{t-1})$, and $gbest = \operatornamewithlimits{arg\,min}_{pbest_i} \mbox{MDL}(pbest_i)$. 

\subsubsection{Mutation}
To expand the search space and help achieving the global optimum faster, ``mutation'' is also employed in the algorithm. In genetics, a mutation is a permanent change to a certain region of a gene. 
In the proposed BPSO, mutation is conducted in the following way: first select the best 10\% of particles (i.e., models with the smallest MDL values). For each selected particle, with equal probability, either (i) modify the number of change points, (ii) adjust the locations of existing change points, or (iii) do both. Finally replace the worst 10\% of particles with the mutated particles. All updates and mutations should satisfy the minimal size constraint. 

\subsubsection{Convergence}
At the end of each iteration, compare the MDL value of the current $gbest$ with the MDL value of the previous $gbest$. The BPSO algorithm terminates if this value is unchanged for five consecutive iterations.

\subsection{Feature Selection}
\label{featselect}
At this stage $\mathcal{B}$, $\mathcal{L}$ and $\mathcal{K}$ are known, and the only parameter left unknown is $\boldsymbol{\beta}$, the linear model coefficients for each region. If $|\mathcal{P}| = P$ and there are $R$ regions, the number of combinations of possible models within all region is $(2^{P+1})^R$ ($P+1$ to include the possibility of an intercept). Heuristically, one can loop over each of these combinations, along with the estimated $\mathcal{B}$, $\mathcal{L}$ and $\mathcal{K}$, and find the combination that results in the smallest MDL value. However, the number of combinations grows exponentially in $P$ and $R$, and soon the calculation will be intractable. Iterative updating methods will be adopted to solve this problem. 

Initially all regions will be assigned to have the full set of $\mathcal{P}$; i.e., $\beta_r = \{\beta_{r,0}, \beta_{r,1}, \dots, \beta_{r,P}\}'$, where $\beta_{r,0}$ is the intercept. At each step, fix all regions but one, say the $r^{th}$ region, and find the model that gives the smallest MDL along with all the other estimated parameter values. Apply this procedure to the remaining regions, and restart again until the models within each region remain unchanged across two large iterations.

\subsection{Final Adjustment}
Since the MDL criterion is non-convex, there is no guarantee that the solution from BPSO is a global minimum solution. A final adjustment is then applied to the final solution from BPSO in the following way. Locate the change points in the final solution from BPSO. For each subset of change points, calculate the corresponding MDL value, as well as the MDL values of some small adjustments of the change points (i.e., adjust the change locations slightly around the located change points). For example, if the data contains four predictors and BPSO locates one change point at the first and third predictor, then there will be four possible subsets of change points. Return the model with the smallest MDL value, and this will be the final solution. Algorithm~\ref{alg:algo} summarizes the above procedure.

\begin{algorithm}
\caption{Minimization of MDL}
\label{alg:algo}
\begin{algorithmic}[1]
\STATE Apply MDL to each predictor to locate a set of candidate change points.
\STATE Initialize particles and velocities $X_i, v_i, i = 1, \dots N$. Set $pbest_i = X_i$ and $gbest = \operatornamewithlimits{arg\,min}_{X_i} \mbox{MDL}(X_i)$. 
\WHILE{BPSO convergence criterion is not met}
    \FOR{$i = 1 \dots N$}
        \STATE Update velocity $v_i$ using (\ref{velocity}) and (\ref{vsigmoid}).
        \STATE Update particle $X_i$ using (\ref{particle}).
        \STATE Conduct variable selection and calculate MDL value.
        \IF{MDL($pbest_i$) $>$ MDL($X_i$)}
                \STATE $pbest_i = X_i$
        \ENDIF
    \ENDFOR
    \STATE Replace the worst 10\% particles with the mutation of the best 10\% particles.
    \STATE Update $gbest = \operatornamewithlimits{arg\,min}_{pbest_i} \mbox{MDL}(pbest_i)$. 
\ENDWHILE
\STATE Apply final adjustment to $gbest$. 
\end{algorithmic}
\end{algorithm}

\section{Empirical Performance}
\label{empiricalstudies}
To evaluate the empirical performance of the proposed methodology, two sets of numerical experiments were conducted. Applications of the proposed methodology to two real data sets were also performed for comparison with a number of existing methods. 

\subsection{Simulation Study: Linear Regression}
\label{regressionsimulation}
The following consists of two examples of partition-wise linear models with regression specifiers, each with two different noise levels. Both examples consist of four predictors, all generated from uniform distributions specified in Table~\ref{regexdist}.
For each setting, two noise levels were used: $\varepsilon \sim N(0,1)$ and $\varepsilon \sim N(0,16)$.  Change points exist at $x_1$ and $x_3$ with $l_1 = l_3 = 1$, and at $x_1$ and $x_4$ with $l_1 = l_4 = 1$, respectively, for the two settings.  The true $\boldsymbol{\beta}$ coefficients are shown in Table~\ref{regexdistbreaks}. All simulations were conducted with 500 trials of $n=200$ and $n=400$ data points for each trial. 

\begin{table}[h]
\caption{Variable distribution for simulations in Section~\ref{regressionsimulation}; $*$ indicates variable with change points, $\dagger$ indicates significant variable in at least one region.}
\label{regexdist}
\vspace*{-0.5cm}
\begin{center}
\begin{tabular}{ccccc}
\hline
& $x_1$ & $x_2$ & $x_3$ & $x_4$ \\
\hline
Setting 1 & $U(0, 7)^{*\dagger}$ & $U(-5,-1)^{\dagger}$ & $U(5, 12)^{*\dagger}$ & $U(-10, -4)^{\dagger}$ \\
Setting 2 & $U(4, 8)^*$ & $U(-5,0)^{\dagger}$ & $U(-9, -3)^{\dagger}$ & $U(0, 3)^*$ 
\end{tabular}
\end{center}
\end{table}

\begin{table}[!htb]
\small
\captionsetup[subtable]{labelformat=empty}
\caption{Change points and coefficients of $\beta$ for simulations in Section~\ref{regressionsimulation}.}
\label{regexdistbreaks}
\vspace*{-0.3cm}
\begin{subtable}{.55\linewidth}
\centering
\caption{Setting 1:}
\begin{tabular}{|c|c|c|c|c|c|}
\hline 
Region & Change Points & $\beta_1$ & $\beta_2$ & $\beta_3$ & $\beta_4$\\ 
\hline 
1 & $x_1 \leq 4, x_3 \leq 8.5$  & 2 & $-2$ & $-4$ & 1 \\ 
\hline 
2 & $x_1 > 4, x_3 \leq 8.5$ & 1.5 & $1$ & 3.5 & $-2$\\ 
\hline 
3 & $x_1 \leq 4, x_3 > 8.5$ & $-1.5$ & $-4.3$ & $-1.7$ & $-2.6$\\ 
\hline 
4 & $x_1 > 4, x_3 > 8.5$ & $-3$ & $-1$ & 2 & 1 \\ 
\hline 
\end{tabular} 
\end{subtable}%
\begin{subtable}{.45\linewidth}
\centering
\caption{Setting 2:}
\begin{tabular}{|c|c|c|c|}
\hline 
Region & Change Points &$\beta_2$ & $\beta_3$ \\ 
\hline 
1 & $x_1 \leq 6, x_4 \leq 1.5$ & $4.2$ & $-4.6$ \\ 
\hline 
2 & $x_1 > 6, x_4 \leq 1.5$ & $-4.2$ & $-4.6$ \\ 
\hline 
3 & $x_1 \leq 6, x_4 > 1.5$ & $4.2$ & $4.6$ \\ 
\hline 
4 & $x_1 > 6, x_4 > 1.5$ & $-4.2$ & $4.6$ \\ 
\hline 
\end{tabular} 
\end{subtable} 
\end{table}

The following results were obtained:

{\bf Setting 1:} For $\varepsilon \sim N(0,1)$, all 500 trials selected the correct sets $\mathcal{B} = \{x_1, x_3\}$ and $\mathcal{L}$ for both $n = 200$ and $n = 400$. Table~\ref{reg12} displays the error distributions between the estimated change points and true change points (of those settings with the correct $\mathcal{B}$ and $\mathcal{L}$). It can be seen that the proposed method estimated the change points correctly with high accuracy. Lastly, results of variable selection within each region is shown in Table~\ref{reg1variable}.  
With a higher noise level $\varepsilon \sim N(0,16)$, 499 and 500 of the trials give the correct sets $\mathcal{B}$ and $\mathcal{L}$ for $n = 200$ and $n = 400$ respectively. The error distributions of the estimated change points and predictor selection results are shown similarly in Tables~\ref{reg12} and~\ref{reg1variable} respectively. One can see that increasing the sample size leads to higher predictor selection accuracy in both settings. 

{\bf Setting 2:} This setting aims to evaluate the situation when the change points occur at variables not significant in any subregions. For $\varepsilon \sim N(0,1)$, 498 trials selected the correct $\mathcal{B} = \{x_1, x_4\}$ and $\mathcal{L}$ for both $n = 200$ and $n = 400$.
With the higher noise level $\varepsilon \sim N(0,16)$, 498 selected the correct sets with $n = 200$, while all all 500 selected correctly with $n = 400$. Tables~\ref{reg12} and~\ref{reg1variable} show the results of this simulation. 

\begin{table}[htbp]
\small
\caption{Mean and standard error of difference between estimated and true change point for simulations in Section~\protect\ref{regressionsimulation}.}
\label{reg12}
\vspace*{-0.3cm}
\centering
\begin{tabular}{|*{6}{c|}} 
\cline{3-6} 
\multicolumn{2}{c|}{}&\multicolumn{2}{c|}{Setting 1} & \multicolumn{2}{c|}{Setting 2}\\
\cline{3-6}
\multicolumn{2}{c|}{}                                           & 
Change Point in $x_1$ &
Change Point in $x_3$ &
Change Point in $x_1$ &
Change Point in $x_4$ \\
\hline
\multirow{2}{*}{$N(0,1)$}&$n = 200$ &$0.0016$ (0.0008) & 0.0000 (0.0000) & $0.0002$ (0.0002) & $0.0000$ (0.0000)\\
&                                                   $n = 400$ &0.0007 (0.0004) & 0.0000 (0.0000) & $0.0000$ (0.0002) & $0.0000$ (0.0000)\\
\hline
\multirow{2}{*}{$N(0,16)$}&$n = 200$ &$0.0084$ (0.0019) &  $0.0000$ (0.0000) & $0.0013$ (0.0007) &  $0.0001$ (0.0001)\\
&                                                       $n = 400$ &$0.0012$ (0.0009) & $0.0000$ (0.0000) & $0.0004$ (0.0004) & 0.0000 (0.0000)\\
\hline
\end{tabular}
\end{table}

\begin{table}[!htb]
\small
\caption{Accuracy of variable selection with each region for simulations in Section~\protect\ref{regressionsimulation}.} 
\label{reg1variable}
\vspace*{-0.3cm}
\centering
\begin{tabular}{|c|c|cccc|cccc|}
\cline{3-10} 
\multicolumn{2}{c|}{}& \multicolumn{4}{c|}{Setting 1} & \multicolumn{4}{c|}{Setting 2}\\
\cline{3-10} 
\multicolumn{2}{c|}{}& Region 1 & Region 2 & Region 3 & Region 4 & Region 1 & Region 2 & Region 3 & Region 4\\
\hline
\multirow{2}{*}{$N(0,1)$}& $n = 200$ & 98\% & 96\% & 97\% & 97\% & 95\% & 94\% & 93\% & 95\%\\ 
&                                                  $n = 400$ & 98\% & 98\% & 99\% & 98\% & 97\% & 98\% & 99\% & 98\%\\ 
\hline 
\multirow{2}{*}{$N(0,16)$}& $n = 200$ & 79\% & 22\% & 71\% & 18\% & 96\% & 95\% & 94\% & 95\%\\
&                                                       $n = 400$ & 97\% & 47\% & 95\% & 53\% & 97\% & 98\% & 99\% & 98\%\\ 
\hline 
\end{tabular}
\end{table}

\subsection{Simulation Study: Classification using Logistic and Probit Regression}
\label{classificationsimulation}
The following two settings assess the accuracy of the proposed method for logistic and probit regression. All settings were simulated under model~(\ref{logisticmodel}) and~(\ref{probitmodel}). Both settings consist of three predictors, and the distribution for each predictor and the $\boldsymbol{\beta}$ coefficients are listed in Tables~\ref{classificationexdist} and~\ref{classificationdistbreaks} respectively. Change points exist at $x_1$ with $l_1 = 2$ for the first setting, and at $x_1$ and $x_3$ with $l_1 = l_3 = 1$ for the second setting. Again all simulations were conducted with 500 trials of $n=200$ and $n=400$ data points each. 

\begin{table}[ht]
\caption{Variable distribution for logistic and probit regression simulations; $*$ indicates variable with change points, $\dagger$ indicates significant variable in at least one region.  Note that in Setting 2 the first predictor is discrete valued.}
\label{classificationexdist}
\vspace*{-0.5cm}
\begin{center}
\begin{tabular}{cccc}
\hline
& $x_1$ & $x_2$ & $x_3$ \\
\hline
Setting 1 & $U(0, 30)^{*\dagger}$ & $U(0,10)^{\dagger}$ & $U(0, 10)$ \\
Setting 2 & $U\{0,\dots,6\}^*$ & $U(0,20)^{\dagger}$ & $U(-10,10)^{*\dagger}$ 
\end{tabular}
\end{center}
\end{table}

\begin{table}[!htb]
\small
\captionsetup[subtable]{labelformat=empty}
\caption{Change points and coefficients of $\beta$ for logistic and probit regression simulation.}
\label{classificationdistbreaks}
\vspace*{-0.3cm}
\begin{subtable}{.5\linewidth}
\centering
\caption{Setting 1:}
\begin{tabular}{|c|c|c|c|c|}
\hline 
Region & Change Points & Intercept & $\beta_1$ & $\beta_2$ \\ 
\hline 
1 & $x_1 \leq 10$ & 0 & 1 & $-1.5$ \\
\hline 
2 & $10 < x_1 \leq 20$ & 0 & 1 & $-4.5$ \\
\hline 
3 & $x_1 > 20$ & 15 & $-1$ & 2\\
\hline
\end{tabular} 
\end{subtable}%
\begin{subtable}{.4\linewidth}
\centering
\caption{Setting 2:}
\begin{tabular}{|c|c|c|c|}
\hline 
Region & Change Points  & $\beta_2$ & $\beta_3$ \\ 
\hline 
1 & $x_1 \leq 3, x_3 \leq 0$ & $2.1$ & 5.1 \\ 
\hline 
2 & $x_1 > 3, x_3 \leq 0$ & 4.0 & $2.4$ \\ 
\hline 
3 & $x_1 \leq 3, x_3 > 0$ & 4.2 & $-5.0$ \\ 
\hline 
4 & $x_1 > 3, x_3 > 0$ &  $-2.9$ & $3.2$ \\ 
\hline 
\end{tabular} 
\end{subtable} 
\end{table}

The following results were obtained:
\begin{itemize}
\item Logistic Regression:\\
{\bf Setting 1:} A total of 464 trials selected the correct set $\mathcal{B}$ and $\mathcal{L}$ for $n = 200$. With increased sample size, the number of trials with correct selection increased to 496. The results are shown in Tables~\ref{classificationresult} and~\ref{classificationvariable} in a similar fashion as before. This setting tests the proposed method's ability on detecting multiple change points at one dimension. Even though the coefficient of $x_1$ did not change across the first two regions, the proposed method still detected the change on the second set of coefficients. 

{\bf Setting 2:} All 500 trials selected the correct $\mathcal{B}$ and $\mathcal{L}$ for both $n = 200$ and $n = 400$. One can see that if a change point occurs at a discrete variable, the proposed method can detect it with very high accuracy. The results are shown in Tables~\ref{classificationresult} and~\ref{classificationvariable} in a similar fashion as before.

\item Probit Regression:\\
{\bf Setting 1:} The results differ only slightly for probit regression. For the smaller sample size $n=200$, 499 of those selected $\mathcal{B}$ and $\mathcal{L}$ correctly. With increased sample size, all 500 trials selected $\mathcal{B}$ correctly. Similar change point and variable selection results for the probit regressions are presented in Tables~\ref{classificationresult} and~\ref{classificationvariable}.

{\bf Setting 2:} All 500 trials selected $\mathcal{B}$ correctly for both $n = 200$ and $n = 400$. As one can see, the results for logistic regression and probit regression are very similar. 
\end{itemize}

\begin{table}[htbp]
\small
\caption{Mean and standard error of difference between estimated and true change point for simulations in Section~\protect\ref{classificationsimulation}.}
\label{classificationresult}
\vspace*{-0.3cm}
\centering
\begin{tabular}{|*{6}{c|}} 
\cline{3-6} 
\multicolumn{2}{c|}{}&\multicolumn{2}{c|}{Setting 1} & \multicolumn{2}{c|}{Setting 2}\\
\cline{3-6}
\multicolumn{2}{c|}{}                                           & Change Point 1 in $x_1$ & Change Point 2 in $x_1$ & Change Point in $x_1$ & Change Point in $x_3$\\
\hline
\multirow{2}{*}{Logistic}&$n = 200$ &$-0.2385$ (0.0564) & $0.0541$ (0.0098) & 0.0000 (0.0000) & 0.0167 (0.0156)\\
&                                                                        $n = 400$ &$-0.0262$ (0.0092) & 0.0164 (0.0038) & 0.0000 (0.0000) & $-0.0017$ (0.0064)\\
\hline
\multirow{2}{*}{Probit}&$n = 200$ &$-0.1179$ (0.0301) & 0.0465 (0.0068) & 0.0000 (0.0000) & 0.0139 (0.0145)\\
&                                                                  $n = 400$ &$-0.0064$ (0.0085) & $0.0154$ (0.0038) & 0.0000 (0.0000) & $-0.0011$ (0.0065)\\
\hline
\end{tabular}
\end{table}

\begin{table}[!htb]
\small
\caption{Accuracy of variable selection with each region for simulations in Section~\protect\ref{classificationsimulation}.} 
\label{classificationvariable}
\vspace*{-0.3cm}
\centering
\begin{tabular}{|c|c|ccc|cccc|}
\cline{3-9} 
\multicolumn{2}{c|}{}& \multicolumn{3}{c|}{Setting 1} & \multicolumn{4}{c|}{Setting 2}\\
\cline{3-9} 
\multicolumn{2}{c|}{}& Region 1 & Region 2 & Region 3 & Region 1 & Region 2 & Region 3 & Region 4\\
\hline
\multirow{2}{*}{Logistic Regression}& $n = 200$ & 91\% & 78\% & 61\% & 86\% & 90\% & 93\% & 92\%\\ 
&                                                  $n = 400$ & 98\% & 93\% & 92\% & 91\% & 92\% & 94\% & 91\%\\ 
\hline 
\multirow{2}{*}{Probit Regrssion}&      $n = 200$ & 89\% & 88\% & 76\% & 96\% & 90\% & 98\% & 96\%\\
&                                                       $n = 400$ & 96\% & 91\% & 95\% & 91\% & 97\% & 97\% & 95\%\\ 
\hline 
\end{tabular}
\end{table}

\subsection{Real Data Analysis: Concrete Compressive Strength Data for Regression}

The construction of high-performance concrete (HPC) relies on not only the common ingredients such as water, cement, fine and coarse aggregates, but also other cementitious materials such as fly ash, superplasticizer, blast furnace slag, etc. However, due to the highly nonlinear relationship between concrete compressive strength and its ingredients, modeling the behavior is extremely difficult. Much work on modeling the strength has been done using artificial neural networks \citep{Yeh98}. In this subsection, the proposed method is used to compare with the prediction performance of neural networks. Along with these two methods, classical regression, regression tree and kernel support vector regression (with radial basis kernel) will be used as well for comparison.

The concrete compressive strength data set \citep{Lichman:2013} consists of 1030 data points, each with 8 predictors, including the variables mentioned above. The goal is to use these predictors to model the concrete compressive strength. 
To evaluate the performance of the methods, a fitted model was first obtained with a training data set consisting of 721 randomly selected observations, and then the remaining 309 observations were used as a testing data set to estimate the prediction error of the fitted model. This process was repeated 100 times (i.e., with 100 different training and testing data sets). The averaged root mean squared prediction error and its standard error are given in Table~\ref{concrete}. All methods (except the proposed one) were trained using the $\texttt{R}$ package $\texttt{caret}$ \citep{caret}. 

\begin{table}[H]
\caption{Averaged root mean squared prediction errors of the concrete strength data set for the five methods considered.  Numbers in parentheses are standard errors.}
\label{concrete}
\begin{center}
\begin{tabular}{ccccc}
\hline
Proposed & Classical Regression & Regression Tree & Neural Networks & Kernel SVR \\ 
\hline 
6.41 (0.04) & 10.52 (0.04) & 12.24 (0.05) & 5.97 (0.04) & 6.85 (0.04)
\end{tabular} 
\end{center}
\end{table}

\begin{figure}
\centering
\includegraphics[scale = 0.36]{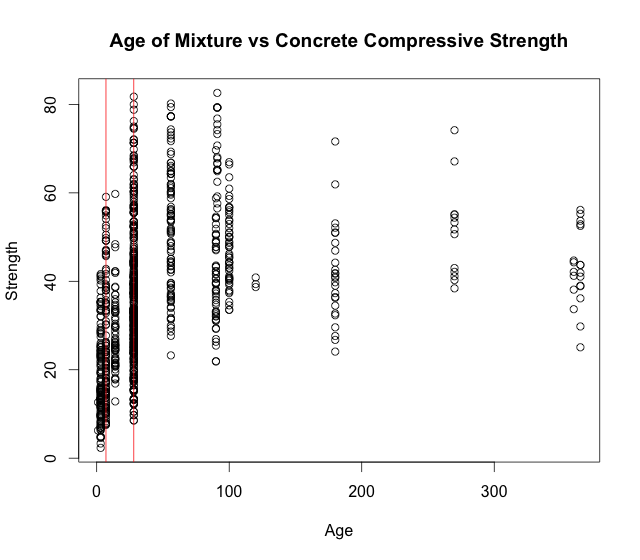}
\includegraphics[scale = 0.36]{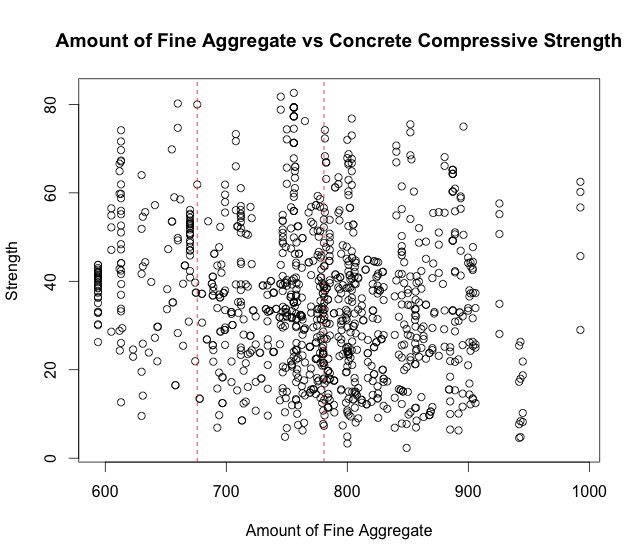}
\caption{Change points for concrete compressive strength data.} 
\label{concretefig}
\end{figure}
From Table~\ref{concrete} one can see the proposed method outperforms most of the competitor methods. It is important to note that even though the proposed method does not perform as well as neural network in terms of prediction, the proposed method captures the relationship between the response and predictors in an easily interpretable way.  Neural network is a highly nonlinear model, thus it can model well the nonlinear behavior of the data, but at the same time the modeled relationship of response and predictors may not be straightforwardly interpreted. The proposed method aims to partition the data space and model each subregion using linear models, therefore one can still interpret the results of the final model. Since rigorous mathematical arguments support the methodology for linear models, statistical inference becomes possible. Figure~\ref{concretefig} shows the change points detected from one of the trials. It shows that both age and fine aggregate variable have two change points, and hence the resulting model partitions the data space into 9 regions. One possible interpretation of the final model is as follows. Compressive strength requires most of the variables to model for high level of fine aggregate, while it only requires at most half of the variables for smaller level of fine aggregate. In general, the variable age has a larger positive influence (larger coefficient value) on the compressive strength for regions with small age values ($< 7$), whereas the effects of all other variables are consistent throughout all regions. 

\subsection{Real Data Analysis: Red Wine Data for Classification}
The red wine data set \citep{Lichman:2013} consists of 961 observations of red Portuguese ``Vinho Verde'' wine. A total of 744 observations are considered bad quality while 217 are considered high quality. Each observation contains 11 predictors. 
The goal is to use these predictors to build a classifier to determine whether a specific set of conditions will lead to good or bad wine quality. The following five existing methods, classical logistic and probit regression, CART, naive bayes and neural networks were used to compare with the performance of the proposed method.  These five classifiers were also trained using the $\texttt{caret}$ package.  The training data set was of size 675, and hence the testing data set was of size 286. As before, for each method the procedure was repeated 100 times and the averaged misclassification error rates (and standard errors) on the testing sets are presented in Table~\ref{wine}.  Here the proposed method seems to give the best performance.


\begin{table}[h]
\caption{Classification error rates for the Red Wine data set.  Numbers in parentheses are standard errors.}
\label{wine}
\vspace*{-0.3cm}
\centering
{\scriptsize
\begin{tabular}{ccccccc}
\hline
proposed (logistic) & proposed (probit) & logistic regression & probit regression & CART & naive bayes & neural networks\\
\hline
0.149 (0.020) & 0.173 (0.021) & 0.204 (0.023) & 0.205 (0.023) & 0.214 (0.024) & 0.225 (0.024) & 0.202 (0.023)
\end{tabular}
}
\end{table}

\section{Conclusion}
\label{conclusion}
This paper presents a new methodology for estimating partition-wise models with linear regression and logistic/probit models as region specifiers.  Under certain distributional assumptions, consistency properties are established for the estimates for the number of change points and their locations. From empirical simulations and real data analysis, there is strong evidence indicating that the proposed method is competitive with existing regression and classification methods. 

\appendix
\section{Appendix}
\section*{Proof of Theorem~\ref{breaklocthm} }
The following section will first present a lemma then the proof of Theorem~\ref{breaklocthm}.
\ignore{
\begin{lemma}
\label{lemma1}
Suppose $(\textbf{x}_i^\prime,y_i)$ follows a linear model $y_i = \textbf{x}'_i\beta + \varepsilon_i$ with $\varepsilon_i \sim N(0, \sigma^2)$, then for $\kappa \in [0,1],$
\[
\frac{1}{n}\sum^{\floor{\kappa n}}_{i=1}\hat{\varepsilon}^2_i \to \kappa \sigma^2
\]
with probability 1 as $n \rightarrow \infty$, where $\hat{\varepsilon}_i = y_i - \textbf{x}'_i\hat{\beta}$, and $\hat{\beta}$ is the maximum likelihood estimate of $\beta$. 
\end{lemma}
\begin{proof}
It is obviously correct if $\kappa = 0$. Let $\kappa \in (0, 1]$. Note that the linear model implies $\hat{\varepsilon}_i - \varepsilon_i = \textbf{x}'_i(\beta - \hat{\beta})$. Also, $\hat{\beta}$ is a strongly consistent estimator for $\beta$. Therefore, with probability one, 
\[
\frac{1}{\floor{\kappa n}}\bigg|\sum^{\floor{\kappa n}}_{i = 1}(\hat{\varepsilon}_i - \varepsilon_i)\bigg| = \frac{1}{\floor{\kappa n}}\bigg|\sum^{\floor{\kappa n}}_{i = 1}\textbf{x}'_i(\hat{\beta} - \beta)\bigg| \to 0.
\]
Therefor, $\frac{1}{\floor{\kappa n}}\sum \limits^{\floor{\kappa n}}_{i = 1} \hat{\varepsilon}_i \to E[\varepsilon_1]$ with probability one by the strong law of large numbers. Since $g(x) = x^2$ is a continuous and measurable function, $\frac{1}{\floor{\kappa n}}\sum\limits^{\floor{\kappa n}}_{i = 1} \hat{\varepsilon}_i^2 \to E[\varepsilon^2_1] = \sigma^2$.
\end{proof}
}
\begin{lemma}
\label{breaksigma}
Suppose $(\textbf{x}_i^\prime,y_i)$ follows a partition-wise linear model (\ref{regionreg}). Then, with probability 1, 
\[
\frac{1}{\hat{n}_r}\sum_{i \in \mathcal{\hat{R}}_r} \hat{\varepsilon}^2_i \to \sigma^2 + B(\mathcal{\hat{R}}_r)
\]
as $n \rightarrow \infty$, where $B(\mathcal{\hat{R}}_r)$ is defined in the proof, and $\mathcal{\hat{R}}_r$ is the partition induced by the change points $\{\mu_{b}, \nu_{b}\}$, $\mu_b < \nu_b$, $\forall$ $b \in \mathcal{B}$.
\end{lemma}
\begin{proof}
Suppose for each $b$, $\mu_{b} \in [\lambda^0_{v_b,b}, \lambda^0_{v_b+1,b})$ and $\nu_b \in (\lambda^0_{v_b',b}, \lambda^0_{v_b'+1,b}]$ for some $0 \leq v_b < v'_b \leq l^0_b$, where $l^0_b$ denotes the true number of change points for the $b^{th}$ predictor, and $\lambda^0_{v_b,b}$ is the true relative location of the $v_b^{th}$ change point for the $b^{th}$ predictor, $1 \leq v_b \leq l^0_b$ (with $\lambda^0_{0,b} = \frac{1}{n}$ and $\lambda^0_{l_b+1,b} = 1$). Note that $\hat{\varepsilon}_i^2$ can be rewritten as $\hat{\varepsilon}_i^2 = (\hat{\varepsilon}_i^2 - \hat{\varepsilon}^2_{i,r}) + \hat{\varepsilon}^2_{i,r}$, $i = 1, \dots, n^0_r, r = 1, \dots, R^0$, where $\hat{\varepsilon}_i$ are the residuals from fit using $\mathcal{\hat{R}}_r$, and $\hat{\varepsilon}_{i,r}$ are the residuals from fit using the true partitions in $\mathcal{R}^0$. Then using the law of large numbers,
\begin{align*}
\frac{1}{\hat{n}_r}\sum_{i \in \mathcal{\hat{R}}_r}\hat{\varepsilon}^2_{i,r} &= \frac{1}{\hat{n}_r}\Big[\sum_{j \in V} \sum_{i \in \mathcal{R}^0_j} \hat{\varepsilon}^2_{i,j} + \sum_{k \in \tilde{V}}\sum_{i \in \mathcal{R}^0_k \cap \mathcal{\hat{R}}_r} \hat{\varepsilon}^2_{i,k}\Big]\\
&\to \frac{1}{\hat{\alpha}_r}\alpha_j\sigma^2 + \frac{1}{\hat{\alpha}_r}\alpha_k \sigma^2 = \sigma^2,
\end{align*}
where $V$ is the subset of $\{1, \dots, R^0\}$ such that $\mathcal{R}^0_j \subset \mathcal{\hat{R}}_r$, and $\tilde{V}$ is the subset of $\{1, \dots, R^0\}\backslash V$ such that $R^0_j \cap \mathcal{\hat{R}}_r \neq \emptyset$, and $\hat{\alpha}_r, \alpha_j, \alpha_k$ are defined as $\floor{\hat{\alpha}_rn} = \hat{n}_r, \floor{\alpha_jn} = \sum \limits_{j \in V} \sum \limits_{i \in \mathcal{R}^0_j}1$, and $\floor{\alpha_kn} = \sum \limits_{k \in \tilde{V}}\sum \limits_{i \in \mathcal{R}^0_k \cap \mathcal{\hat{R}}_r} 1$ respectively. Now define $\Delta \varepsilon^2_{i, r} = \hat{\varepsilon}_i^2 - \hat{\varepsilon}^2_{i,r}$. Note that each of these averages will converge to a nonzero bias term, i.e.
\[
\frac{1}{\hat{n}_r}\sum_{i \in \mathcal{\hat{R}}_r}\Delta \hat{\varepsilon}^2_{i,r} = \frac{1}{\hat{n}_r}\Big[\sum_{j \in V} \sum_{i \in \mathcal{R}^0_j}\Delta \hat{\varepsilon}^2_{i,j} + \sum_{k \in \tilde{V}}\sum_{i \in \mathcal{R}^0_k \cap \mathcal{\hat{R}}_r}\Delta  \hat{\varepsilon}^2_{i,k}\Big] \to B(\mathcal{\hat{R}}_r),
\]
where $B(\mathcal{\hat{R}}_r)$ is a bias term whose exact form depends on the true segmentation. Moreover, $B(\mathcal{\hat{R}}_r) > 0$ unless all the $\{\mu_b, \nu_b\}$ coincide with the true change point locations. If all the estimated locations coincide with the true change point locations, then 
\[
\frac{1}{\hat{n}_r}\sum_{i \in \mathcal{\hat{R}}_r} \hat{\varepsilon}^2_i \to \sigma^2,
\]
and this completes the proof.
\end{proof}

\begin{proof}[\bf Proof of Theorem~\ref{breaklocthm}]
Suppose the true partitions are denoted as $\mathcal{R}^0 = \{\mathcal{R}^0_1, \dots, \mathcal{R}^0_{R^0}\}$, and the estimated partitions by the MDL criterion (\ref{regestimate}) as $\mathcal{\hat{R}} = \{\mathcal{\hat{R}}_1, \dots, \mathcal{\hat{R}}_{R^0}\}$. Suppose as $n \rightarrow \infty$, $\mathcal{\hat{R}} \not \to \mathcal{R}^0$ with probability 1. Then by boundedness, there exists a subsequence along which $\mathcal{\hat{R}}$ converges to with probability 1, say $\mathcal{\hat{R}} \to \mathcal{R}^* \neq \mathcal{R}^0$. Then, for each $r' \in \{1, \dots, R^0\}$, either 1) $\mathcal{R}^*_{r'} \subset \mathcal{R}^0_r$ for some $r$, or 2) $\mathcal{R}^*_{r'} = \cup_r (\mathcal{R}^*_{r'} \cap \mathcal{R}^0_r)$. Note that with probability 1, $\frac2n\mbox{MDL}(R^0, \mathcal{R}^*) \sim \log(\frac1n\mbox{RSS}^*_{R^0})$, where for two sequences $a_n$ and $b_n$, $a_n \sim b_n$ if $\lim\limits_{n \to \infty}\frac{a_n}{b_n} = 1$, and $\mbox{RSS}^*_{R^0} = \sum\limits_{\mathcal{R}^*_r \in \mathcal{R}^*}\sum\limits_{i \in \mathcal{R}^*_r} (y_i - \hat{f}_r(x_i))^2 = \sum\limits_{\mathcal{R}^*_r \in \mathcal{R}^*}\sum\limits_{i \in R^*_r} \hat{\varepsilon}_i^2$. Now, for case one,
\[
\frac{1}{n} \sum_{i \in \mathcal{R}^*_{r'}}\hat{\varepsilon}^2_i \to \alpha_{r'}^* \sigma^2
\]
by an application of law of large numbers, where $\floor{\alpha^*_{r'} n} = |R^*_{r'}|$. For case two, Lemma~\ref{breaksigma} implies
\[
\lim_{n\to\infty} \frac{1}{n}\sum_{i \in \mathcal{R}^*_{r'}} \hat{\varepsilon}^2_i = \alpha^*_{r'}\sigma^2 + B(\mathcal{R}^*_{r'}) > \alpha^*_{r'}\sigma^2.
\]
Together, with the concavity of $\log$, 
\begin{align*}
\lim_{n \to \infty} \frac{2}{n}\mbox{MDL}(R^0, \mathcal{R}^*) &= \lim_{n \to \infty} \log\Big(\frac{1}{n}\sum\limits_{\mathcal{R}^*_r \in \mathcal{R}^*}\sum\limits_{i \in \mathcal{R}^*_{r'}} (y_i - \hat{f}_r(x_i))^2\Big) \\
&> \lim_{n \to \infty} \log\Big(\frac{1}{n}\sum\limits_{\mathcal{R}^0_r \in \mathcal{R}^0}\sum\limits_{i \in \mathcal{R}^0_r} (y_i - \hat{f}_r(x_i))^2\Big)\\
&= \lim_{n \to \infty}\frac{2}{n}\mbox{MDL}(R^0, \mathcal{R}^0) \geq \lim_{n \to \infty} \frac{2}{n}\mbox{MDL}(R^0, \mathcal{R}^*),
\end{align*}
which is a contradiction. Hence $\hat{\mathcal{R}} \to \mathcal{R}^0$ with probability 1. The proof for the logistic/probit model follows similarly.
\end{proof}

\section*{Proof of Theorem~\ref{breakconsistthm}}
The following section will present five lemmas along with the proof of Theorem~\ref{breakconsistthm}.
\begin{lemma}
\label{morebreak}
Suppose $\hat{\mathcal{L}}$ is an estimator from (\ref{regestimate}). Then $\forall$ $\hat{l}_b \in \hat{\mathcal{L}}$, $P(\hat{l}_b \geq l_b^0) \rightarrow 1$ with probability 1 as $n \rightarrow \infty$.
\end{lemma}
\begin{proof}
Consider the following two cases:
\begin{itemize}
\item Case 1: $\hat{l}_b < l^0_b$ $\forall$ $b \in \mathcal{B}$. Then $\hat{R} < R^0$. By Lemma~\ref{breaksigma}, there exists a partition region $\hat{\mathcal{R}}_r$ such that it intersects at least two or more true regions $\mathcal{R}^0_j$. Thus by Lemma~\ref{breaksigma}, $\frac1nRSS(\hat{\mathcal{R}}) \rightarrow \sigma^2 + B(\hat{\mathcal{R}})$. 
\item Case 2: $\hat{l}_b < l^0_b$ for some $b \in \mathcal{B}$. If $\prod \limits_{b \in \mathcal{B}} (\hat{l}_b + 1) = \hat{R} < R^0 = \prod \limits_{b \in \mathcal{B}} (l_b^0 + 1)$, then following case 1, $\frac1nRSS(\hat{\mathcal{R}}) \rightarrow \sigma^2 + B(\hat{\mathcal{R}})$. If $\hat{R} \geq R^0$, this implies for some $b \in \mathcal{B}$ there are too many pieces. Thus there still exists at least one $\hat{\mathcal{R}}_r$ that intersects at least two or more $\mathcal{R}^0_j$. Together with case 1 the claim holds. \qedhere
\end{itemize}
\end{proof}

\begin{lemma}
\label{split}
For each $1 \leq v_b \leq l^0_b$, $l^0_b < l_b \leq L_b$ and $b \in \mathcal{B}$, define the sets $\mathcal{A}_{v_b}(n) = \{(k_{1,b}, \dots, k_{l_b,b}):0 < k_{1,b} < \dots < k_{l_b,b} < n, |k_{s,b} - k^0_{v_b,b}| \geq [\log n]^2$ for all $1 \leq s \leq l_b\}$. Denote $\mathcal{A}_V = \{\mathcal{A}_{v_1}(n), \dots, \mathcal{A}_{v_B}(n)\}$ for any combinations of $V = \{v_1, \dots, v_B\}$, then 
\[
P(\hat{\mathcal{K}} \in \mathcal{A_V}) \rightarrow 0
\]
as $n \rightarrow \infty$. 
\end{lemma}

\begin{proof}
Let $\mathcal{K} \in \mathcal{A}_{V}$, and $\mathcal{R}$ be the set of partitions induced by $\mathcal{K}$. Define $\tilde{\mathcal{R}}$ be the set of partitions induced by the change points $\tilde{\mathcal{K}}_b = \{k_{1,b}, \dots, k_{l_b,b}, k^0_{1,b}, \dots, k^0_{v_b-1,b}, k^0_{v_b,b}-[\log n]^2, k^0_{v_b,b}+[\log n]^2, k^0_{v_b+1,b}, \dots, k^0_{l^0_b,b}\}$ for all $b \in \mathcal{B}$. It is obvious that $RSS(\mathcal{R}) \geq RSS(\mathcal{\tilde{R}})$. Following the construction in \citet{yao88}, $\tilde{\mathcal{K}}_b$ can be decomposed the following way. For each $b \in \mathcal{B}$, define $U_{s,b}$ for $s = 1, \dots, v_b-1, v_b+2, \dots, l_b^0+1$ be the set of points $i \in (k^0_{s-1,b}, k^0_{s,b}]$, $U_{v_b,b}$ the set of points $i \in (k^0_{v_b-1,b}, k^0_{v_b} - [\log n]^2 ]$, $U_{v_b+1,b}$ with the points $i \in (k^0_{v_b,b} + [\log n]^2, k^0_{v_b + 1,b}]$, and $U_{l^0_b+2, b}$ the set of points $i \in (k^0_{v_b,b} - [\log n]^2, k^0_{v_b,b} + [\log n]^2]$. Denote $\mathcal{U}_b = \{U_{1,b}, \dots, U_{l^0_b+2,b}\}$, and $\bar{\mathcal{R}} = \bar{\mathcal{R}}_1 \cup \bar{\mathcal{R}}_2$ be the set of partitions induced by $\mathcal{U}_1, \dots, \mathcal{U}_B$, where $\bar{\mathcal{R}}_1$ contains the partition with all boundaries formed by $U_{l^0_b+2,b}$ $\forall$ $b \in \mathcal{B}$, and $\bar{\mathcal{R}}_2$ contains all other partitions. Following \citet{yao88} and \citet{AueLee11}, 
\[
0 \leq \sum_{\bar{\mathcal{R}}_{v,2} \in \bar{\mathcal{R}}_2}\sum_{i \in \bar{\mathcal{R}}_{v,2}} \varepsilon^2_i - \sum_{\bar{\mathcal{R}}_{v,2} \in \bar{\mathcal{R}}_2}\sum_{i \in \bar{\mathcal{R}}_{v,2}}(y_i - X_i'\hat{\beta}(\bar{\mathcal{R}}_{v,2}))^2 = O_p(\log n).
\]
For $\bar{\mathcal{R}}_1$, first note that $\bar{n} = |\bar{\mathcal{R}}_1| \propto [\log n]^2$. Let $\mathcal{I} = \{r: \mathcal{R}^0_r \cap \bar{\mathcal{R}}_1 \neq \emptyset\}$, and for any $r = 1, \dots, R^0$, use the relationship $y_i - X_i'\hat{\beta}(\bar{\mathcal{R}}_1) = X_i'\beta_r + \varepsilon_i - X_i'\hat{\beta}(\bar{\mathcal{R}}_1)$ and obtain

\begin{align}
\frac{1}{\bar{n}}\big(\sum_{i \in \bar{\mathcal{R}}_1}\varepsilon^2_i - \sum_{i \in \bar{\mathcal{R}}_1} \hat{\varepsilon}^2_i \big) &= \frac{1}{\bar{n}}\big(\sum_{r \in \mathcal{I}}\sum_{i \in \mathcal{R}^0_r \cap \bar{\mathcal{R}}_1}\varepsilon^2_i - \sum_{r \in \mathcal{I}}\sum_{i \in \mathcal{R}^0_r \cap \bar{\mathcal{R}}_1} (y_i - X_i'\hat{\beta}(\bar{\mathcal{R}}_1))^2\big) \nonumber \\
&\approx-\frac{1}{\bar{n}} \sum_{r \in \mathcal{I}}\sum_{i \in \mathcal{R}^0_r \cap \bar{\mathcal{R}}_1} (X_i'\beta_v - X_i\hat{\beta}(\bar{\mathcal{R}}_1))^2 \nonumber \\
&= -\frac{1}{\bar{n}} \sum_{r \in \mathcal{I}}\sum_{i \in \mathcal{R}^0_r \cap \bar{\mathcal{R}}_1} (X_i'd_r)^2 \rightarrow B < 0 \nonumber 
\end{align} 

as $n \rightarrow \infty$, with $\hat{\beta}(\bar{\mathcal{R}}_1)$ be the estimate of $\beta$ using the observations in $\bar{\mathcal{R}}_1$. Thus, together gives 
\[
\frac{1}{\bar{n}} (RSS - RSS(\tilde{R})) < 0,
\]
which leads to 
\[
\min_{\mathcal{R}} RSS(\mathcal{R}) \geq RSS(\tilde{R}) > RSS \geq RSS(\hat{\mathcal{R}})
\]
with probability 1, which is a contradiction (note $RSS = \sum \limits_{i = 1}^n \varepsilon_i^2$).
\end{proof}

\begin{lemma}
\label{mdls}
Suppose $l_b \geq l^0_b$ $\forall$ $b\in \mathcal{B}$, then the following statements hold with probability approaching 1 as $n \rightarrow \infty$:
\begin{enumerate}[i)]
\item $\sum \limits^{l_b+1}_{z = 1}\log_2n_{z,b} - \sum \limits^{l^0_b+1}_{z = 1}\log_2n^0_{z,b} \geq 0$
\item $\sum \limits_{\mathcal{R}_r \in \mathcal{R}} \log_2R - \sum \limits_{\mathcal{R}^0_r \in \mathcal{R}^0} \log_2R^0 \geq 0$
\item $\sum \limits_{\mathcal{R}_r \in \mathcal{R}} \frac{s_r}{2}\log_2n_r - \sum \limits_{\mathcal{R}^0_r \in \mathcal{R}^0} \frac{s_r^0}{2}\log_2n_r^0 \geq 0$.
\end{enumerate}
\end{lemma}

\begin{proof}
\hfill
\vspace{-0.2in}
\begin{enumerate}[i)]
\item Denote $\{P^0_{z',b}\}_{z' = 1}^{l^0_b+1}$ the pieces at each dimension $b \in \mathcal{B}$ and the estimated pieces $\{P_{z,b}\}_{z = 1}^{l_b+1}$. Lemma \ref{split} implies that for each $P_{z,b}$, there exists a $P^0_{z',b}$ such that $P_{z,b} \subseteq P^0_{z',b}$ as $n \rightarrow \infty$. Let $n_{z,b} = \alpha_{z,b}n$ for $z = 1, \dots, l_b+1$ and $n_{z,b}^0 = \alpha_{z,b}^0n$ for $z = 1, \dots, l^0_b+1$, then there exists $\gamma_{z, z', b}$ such that $\gamma_{z, z', b} \rightarrow \frac{\alpha_{z,b}}{\alpha^0_{z',b}}$. Denote $\mathcal{J}_{z'}^b = \{z:P_{z,b} \subseteq P^0_{z',b}, z = 1, \dots, l_b+1\}$ for $z' = 1, \dots, l^0_b$. Then,
\begin{align}
\Big[\prod^{l_b+1}_{z = 1}n_{z,b}\Big]\Big[\prod^{l_b^0+1}_{z' = 1}n^0_{z',b}\Big]^{-1} &= \Big[\prod^{l^0_b+1}_{z' = 1}\prod_{z \in \mathcal{J}_{z'}^b}\gamma_{z, z', b}(n^0_{z',b})^{|\mathcal{J}_{z'}^b|}\Big]\Big[\prod^{l^0_b+1}_{z=1}n^0_{z',b}\Big]^{-1} \nonumber \\
&= \Big[\prod^{l^0_b+1}_{z' = 1}\prod_{z \in \mathcal{J}_{z'}^b}\gamma_{z, z', b}(n^0_{z',b})^{|\mathcal{J}_{z'}^b|-1}\Big] \nonumber \\
&\geq  (\min n^0_{z',b})^{l_b - l_b^0}\prod^{l^0_b+1}_{z' = 1}\prod_{z \in \mathcal{J}_{z'}^b}\gamma_{z, z', b} \nonumber \\
& \geq  1 \nonumber
\end{align}
as $n \rightarrow \infty$ since $\prod^{l^0_b+1}_{z' = 1}\prod_{z \in \mathcal{J}_{z'}^b}\gamma_{z, z', b}$ converges to a finite number lower bounded away from 0. Taking $\log$ gives desired result.
\item If $l_b > l_b^0$ for some $b \in \mathcal{B}$, then $R > R^0$. The claim then follows.
\item Suppose for now assume that $s_r = s^0_{r'} = s$. Then the claim follows from the proof of part i). The proof of $s_r \neq s^0_{r'}$ is not of focus in this context, and will be proved in future work. \qedhere
\end{enumerate}
\end{proof}

\begin{lemma}
\label{bound}
Suppose $l_b \geq l_b^0$ $\forall$ $b \in \mathcal{B}$ and strict inequality for some $b$. Then,
\[
P(0 \leq RSS - RSS(\hat{\mathcal{R}}) < Q(\hat{\mathcal{L}},\epsilon)) \rightarrow 1,
\]
where $Q(\hat{\mathcal{L}}, \epsilon) = \sigma^2\log n \{\epsilon + 2[\prod \limits_{b \in \mathcal{B}}(\hat{l}_b - l_b^0)-1](1+\epsilon)\}$.
\end{lemma}

\begin{proof}
Following the notations in the proof of Lemma \ref{split}, let $\mathcal{W}(n)$ be the intersection of all the complements of $A_V(n)$. Due to Lemma \ref{split}, it is enough to proof the claim for any $\mathcal{R}_w$ induced by the change locations $w \in \mathcal{W}(n)$. For any $w$, $RSS(\mathcal{R}_w) \geq RSS(\tilde{\mathcal{R}}_w)$, where $\tilde{\mathcal{R}}_w$ is induced by the change locations $\mathcal{M}_{w,b} = \{k_{1,b}, \dots, k_{l_b,b}, k^0_{1,b}, \dots, k^0_{l^0_b,b}, k^0_{1,b}-[\log n]^2, \dots, k^0_{l^0_b,b}-[\log n]^2, k^0_{1,b}+[\log n]^2, \dots, k^0_{l^0_b,b}+[\log n]^2\}$ for each $b \in \mathcal{B}$. Following \citet{yao88}, define $V_{s,1,b}$ be the set of points $i \in (k^0_{s,b} - [\log n]^2, k^0_{s,b}]$ for $s = 1, \dots, l^0_b$, $V_{s,2,b}$ the set of points $i \in (k^0_{s,b}, k^0_{s,b} + [\log n]^2]$ for $s = 1, \dots, l^0_b$, $V_{1,3,b}$ the set of points $i \in (0, k^0_{1,b} - [\log n]^2]$, $V_{l^0_b+1,3,b}$ the set of points $i \in (k^0_{l^0_b,b} + [\log n]^2, n]$, and $V_{s,3,b}$ the set of points $i \in (k^0_{s-1,b} + [\log n]^2, k^0_{s,b} - [\log n]^2]$. Then $RSS(\tilde{R}_w)$ can be decomposed into the sum of RSS induced by these new regions, denote $\mathcal{R}^{\prime}$. Note that $\mathcal{R}^{\prime}$ contains three types of regions: i) none of the boundaries are induced by points from $V_{\cdot,3,b}$, ii) some boundaries of $\mathcal{R}^{\prime}$ are induced by points from $V_{\cdot, 3,b}$, and iii) the all boundaries of $\mathcal{R}^{\prime}$ are induced by points from $V_{\cdot,3,b}$. For case i, an application of Lemma 1 from \citet{yao88} gives 
\[
\sum_{i \in \mathcal{R}^{\prime}_r}\varepsilon^2_i - \sum_{i \in \mathcal{R}^{\prime}_{r,i}}\hat{\varepsilon}^2_i = O_p(\log \log n).
\]
For case ii, note that the number of data points in each of those regions are still bounded by $\log^2n$, thus the result
\[
\sum_{i \in \mathcal{R}^{\prime}_r}\varepsilon^2_i - \sum_{i \in \mathcal{R}^{\prime}_{r,ii}}\hat{\varepsilon}^2_i = O_p(\log \log n)
\]
still holds. For case iii, note that the number of these regions $\mathcal{R}^{\prime}_{r,iii}$ that contains some $k_{\cdot,b}$ on at least one boundary is bounded by $\prod\limits_b (l_b - l_b^0)$. Thus as in Lemma 5 of \citet{yao88},
\[
\sum_{i = 1}^n \epsilon^2_i \geq RSS(\mathcal{R}_w) \geq \sum_{i=1}^n \varepsilon^2_i - Q(w, \epsilon),
\]
and this completes the proof.
\end{proof}

\begin{lemma}
\label{order}
Suppose $\hat{l}_b \geq l_b$ $\forall$ $b \in \mathcal{B}$, and let $\hat{\mathcal{R}}$ be the partition induced by the changes, then 
\[
P\bigg(\frac{n}{2}\Big[\log\Big(\frac{RSS(\hat{\mathcal{R}})}{n}\Big) - \log\Big(\frac{RSS(\mathcal{R}^0)}{n}\Big) \Big] \geq 0 \bigg ) \rightarrow 1
\] 
as $n \rightarrow \infty$.
\end{lemma}

\begin{proof}
Note that as $n \rightarrow \infty$, $RSS = \sum \limits_{i=1}^n \varepsilon^2 > n(\sigma^2 - \epsilon)$ for $\epsilon > 0$. Also, note that $RSS \geq RSS(\mathcal{R}^0)$. Therefore,
\begin{align}
\label{RSS}
\frac{n}{2}\Big[\log\Big(\frac{RSS(\hat{\mathcal{R}})}{n}\Big) - \log\Big(\frac{RSS(\mathcal{R}^0)}{n}\Big) \Big] &\geq  \frac{n}{2}\Big[\log\Big(\frac{RSS(\hat{\mathcal{R}})}{n}\Big) - \log\Big(\frac{RSS}{n}\Big) \Big] \nonumber \\
&= \frac{n}{2} \log\Big(1 - \frac{RSS - RSS(\hat{\mathcal{R}})}{RSS}\Big) \nonumber \\
&\geq  \frac{n}{2}\log \Big(1 - \frac{Q(\hat{\mathcal{R}},\epsilon)}{n(\sigma^2 - \epsilon)}\Big). 
\end{align}
Using the inequality $\log(1-x) > -x(1+\epsilon)$ for small $x > 0$, the RHS of (\ref{RSS}) is greater than 
\[
-\frac{(1+\epsilon)}{2(\sigma^2 - \epsilon)}\sigma^2\log n\{\epsilon + 2[\prod \limits_{b \in \mathcal{B}}(\hat{l}_b - l_b^0) - 1](1 + \epsilon)\}
\]
with probability 1 for positive for small $\epsilon$. This completes the proof of the claim. 
\end{proof}
Combining with Lemma \ref{mdls}, this completes the proof that $\hat{\mathcal{L}} \xrightarrow{P} \mathcal{L}^0$. The second part follows from the fact that $P(\hat{\mathcal{R}}) \geq P(\hat{\mathcal{R}}, \hat{\mathcal{L}} = \mathcal{L}^0) \rightarrow 1$. 

\section*{Proof of Corollary 1}
\begin{proof}
The proof for the probit model follows by rewriting (\ref{glm}) as a latent variable model, and noticing the one-to-one relationship between the latent variable and the observed value. In particular, for each observation $\textbf{x}_i$, suppose there is a latent variable $y_i^*$ such that
\[
y_i^* = \textbf{x}_i\beta + \varepsilon_i,
\]
where $\varepsilon_i \sim N(0,1)$. Then the observations $y_i$ can be determined as 
\[
y_i = \left\{
\begin{array}{cl}
1, & y_i^* > 0 \nonumber \\
0, & otherwise. \nonumber
\end{array}
\right.
\qedhere \]

\end{proof}

\bibliographystyle{agsm}
\bibliography{tleeref}
\end{document}